\documentclass[11pt]{article}

\usepackage[english]{babel}
\usepackage[utf8x]{inputenc}
\usepackage{libertine}
\usepackage{amsmath}
\usepackage{amsthm, amssymb}
\usepackage{amsfonts}
\usepackage{algorithm}
\usepackage{algorithmic}
\usepackage{graphicx}
\usepackage[dvipsnames]{color}
\usepackage{fullpage}
\usepackage{bbm}
\usepackage[
  pagebackref,
  colorlinks=true,
  urlcolor=blue,
  linkcolor=blue,
  citecolor=blue,
]{hyperref}
\usepackage[nameinlink]{cleveref}

\newtheorem{theorem}{Theorem}

\newtheorem{definition}[theorem]{Definition}
\newtheorem{lemma}[theorem]{Lemma}
\newtheorem{corollary}[theorem]{Corollary}

\newtheorem{example}{Example}

\newcommand{\eps}{\varepsilon}
\newcommand{\E}{\mathbb{E}}

\newcommand{\allocsub}[1]{\mathbb{A}_{#1}}
\newcommand{\allocsubi}{\allocsub{i}}

\newcommand{\inspect}[1]{\mathbb{I}_{#1}}
\newcommand{\inspecti}{\inspect{i}}

\newcommand{\unopened}{\mathcal{U}}
\newcommand{\M}{\mathcal{M}}
\newcommand{\A}{\mathcal{A}}
\newcommand{\W}{\mathcal{W}}
\newcommand{\Pol}[1]{\mathcal{P}^{#1}}
\newcommand{\Poli}{\Pol{i}}
\newcommand{\Pols}{\Pol{S}}
\newcommand{\strike}{\sigma}
\newcommand{\covered}{\kappa}
\newcommand{\R}{\mathbb{R}}

\newcommand{\Prob}{\mathbb{P}}
\newcommand{\indic}[1]{{\mathbf{1}_{#1}}}
\newcommand{\NonAdapt}{\mbox{\sc NonAdapt}}

\title{Pandora's Problem with Nonobligatory Inspection}
\author{Hedyeh Beyhaghi and Robert Kleinberg\thanks{Supported in part by
NSF grant CCF-1512964.} \\ {\em Cornell University}}
\date{}
\begin{document}

\maketitle

\begin{abstract}
  Martin Weitzman's ``Pandora's problem'' furnishes the
  mathematical basis for optimal search theory in economics.
  Nearly 40 years later, Laura Doval introduced a version
  of the problem in which
  the searcher is not obligated to pay the cost of inspecting
  an alternative's value before selecting it.
  Unlike the original Pandora's problem,
  the version with nonobligatory inspection
  cannot be solved optimally by any simple ranking-based policy,
  and it is unknown whether there exists any polynomial-time
  algorithm to compute the optimal policy. This motivates
  the study of approximately optimal policies that are
  simple and computationally efficient. In this work we
  provide the first non-trivial approximation guarantees
  for this problem. We introduce a family of
  ``committing policies'' such that it is computationally
  easy to find and implement the optimal committing policy.
  We prove that the optimal committing policy is guaranteed
  to approximate the fully optimal policy within a
  $1-\frac1e = 0.63\ldots$ factor, and for the special
  case of two boxes we improve this factor to $4/5$ and
  show that this approximation is tight for the class of
  committing policies.
\end{abstract}

\maketitle

\section{Introduction}
\label{sec:intro}

Search theory, which concerns the ways in which costs of
obtaining information affect the structure and
outcome of optimization procedures, was born
in 1961 when the economist
George Stigler~\cite{stigler} sought to understand
the phenomenon of price dispersion.
When sellers charge different prices for
identical goods, why do consumers ever
choose the higher-priced seller? Stigler
realized that this counter-intuitive
behavior could be explained by {\em search
frictions} whereby consumers must expend
costly effort to find and/or evaluate sellers.
%

The insight that optimization has qualitatively
different outcomes under search frictions
resounded beyond economics, and particularly
within computer science. Models of
costly information acquisition have been
incorporated into information
retrieval, robotics, database theory,
distributed systems, and of
course also into sub-areas of CS
such as algorithmic pricing and mechanism design
that explicitly relate to economics.

From a mathematical standpoint, the most
foundational model of optimal search
was articulated by Martin Weitzman~\cite{weitzman}
under the name {\em Pandora's
problem}. 
The basic elements of the problem are as
follows. A searcher is allowed to select
a prize from one of $n$ closed boxes.
The values of the prizes inside the boxes
are independent random variables, sampled
from (not necessarily identical) distributions
that are known to the searcher.
The searcher chooses a sequence of operations,
each of which
is either {\em opening} a box or
{\em selecting} the box. Opening box $i$
has an associated cost $c_i$ and results
in learning the value $v_i$ of the prize
contained inside. Selecting box $i$
results in a payoff of $v_i$ and immediately
ends the search process; this operation can
only be performed after box $i$ has been opened.
The searcher's goal
is to design an adaptive policy
(i.e., a choice of which operation to
perform next, for every possible past
history of operations and their outcomes)
to maximize the expectation of the prize
selected, minus the sum of the inspection
costs accrued while opening boxes.

A priori, it would appear
that the solution to Pandora's problem may
be horribly complex. An optimal adaptive
policy must specify the next operation to
be performed given any past history. If
each $v_i$ is drawn from a distribution
with support size $s$, the number of
possible histories is $s^n$,
so adaptive policies in general have
exponential description size. It is
easy to see that the optimal policy
can be implemented in space
$\operatorname{poly}(n,s)$,
but there is no obvious reason why
the complexity of Pandora's problem
should lie anywhere below PSPACE.

Surprisingly, though, the solution to
Pandora's problem is not complex at all.
Weitzman proved that the
optimal policy
has a beautifully
simple structure: one computes a
{\em reservation value} for each box, sorts
them in decreasing order of reservation value,
and opens them in this order, stopping
and selecting the first open box
whose prize value exceeds the reservation
value of every remaining closed box.
This entire process can be implemented
to run in time $O(ns \log(ns))$.

A key assumption in Pandora's problem
is that the searcher must open
a box, and suffer the attendant cost,
before selecting it. This assumption
limits the applicability of the model.
Given that the value of the prize inside a box
is drawn from a distribution known in
advance, in many cases it may be more
advantageous to select a box without
paying to inspect its contents. For
example, when using Pandora's problem
to model a firm searching for
an employee to hire, boxes represent
job candidates. The cost of opening
a box represents the cost to the firm of undertaking
a process, such as an interview or internship,
to assess the value of hiring a candidate.
If the evidence of a candidate's promise is sufficiently
strong {\em a priori}, it may be realistic to
assume that the firm is willing to hire him or her
directly, skipping the costly evaluation
process. This motivates a version of Pandora's
problem in which a box may be selected without
opening it, if the searcher so desires.

Given that Weitzman's original model dates
from 1979 and has been cited almost 900 times,
it is quite surprising that this close variant
never appeared in the literature until a 2018 paper
by Laura Doval~\cite{doval}.
The relative unpopularity of the variant with
nonobligatory inspection can probably be attributed
to the apparent complexity and lack of structure
in its optimal solution. For example, Doval
presents an example of a problem instance
(Problem 3 in~\cite{doval})
with three boxes --- A, B, and C --- such that
the optimal policy first opens box C, but the
question of whether it subsequently
opens box A before B or vice-versa depends
on the value of the prize discovered inside box C.
As before, one can easily show that  this variant
of Pandora's problem belongs to PSPACE, but
unlike Weitzman's version of Pandora's problem,
there is no evidence that this
version is easier than PSPACE-complete.

These considerations motivate the study of
approximately optimal policies that are
computationally efficient, structurally
simple, or both. Our work initiates this
study.

\subsection{Our results and techniques}
\label{sec:results}

To put our results in context, we begin
this section with an easy observation
showing that simple, computationally
efficient policies can attain at least a
$\frac12$-approximation to the optimal
policy. Consider the following two
policies.
\begin{enumerate}
\item \label{item:weitzman}
{\bf [Policy A]} Run
Weitzman's optimal policy, ignoring the
fact that the searcher has the option to
select boxes without opening them.
\item \label{item:best-closed}
{\bf [Policy B]} Leave every box closed,
and select the one with the highest expected value.
\end{enumerate}
Among all adaptive policies, Policy A is the
one that maximizes
the expected net contribution (i.e., the value if
selected, minus the inspection cost) of
open boxes, whereas Policy B
maximizes the expected net contribution
of closed boxes.
Hence the combined value of Policies A and B
bounds from above the combined value that
the optimal policy obtains from both
open and closed boxes. The better
of A and B must consequently attain at least half the
value of the optimal policy.

For any specified $\eps>0$,
it is not hard to construct
a problem instance such that
neither Policy A nor Policy B
attains more than $\frac12 + \eps$
of the value of the optimal policy.
To achieve a better approximation
factor, we focus on a broader class
of simple policies that includes
both of the aforementioned ones.

Let us define a {\em committing policy} to be one
that, before it opens any boxes,
must pre-commit to a partition of the
$n$ boxes into a set of boxes that will
never be opened and a set that will never
be selected without first being opened; in
addition it pre-commits to an order in which
the boxes in the latter set will be opened.
Such a policy is almost non-adaptive;
the only way in which it may adjust its
behavior in response to information revealed
during the search process is that it may
terminate the search early. In this
sense, questions about the ability of
 committing policies to approximate the
optimal adaptive policy are akin to
questions about {\em adaptivity gaps}
in stochastic optimization~\cite{an,dgv,gns16,gns17}.

The foregoing discussion inspires two
interrelated questions.
\begin{description}
\item[Question 1] {\em For which values of
$\alpha$ is there a polynomial-time algorithm
that $\alpha$-approximates the optimal adaptive
policy?}
\item[Question 2] {\em What is the worst-case
ratio between the value of the optimal  committing
policy and that of the optimal adaptive policy?}
\end{description}
We show that for the general case of Pandora's
problem with nonobligatory inspection, there is
a polynomial-time algorithm to identify the
optimal  committing policy, and this policy
always attains at least $1 - \frac1e = 0.63\ldots$
fraction of the optimal policy's value.
This furnishes a non-trivial lower bound on the
answers to Questions 1 and 2 above. Our
second main result fully settles Question 2
for the case of two boxes: we show that
the optimal  committing policy is always
at least a $\frac45$-approximation to
the optimal adaptive policy, and that
this approximation factor is tight.
The main question left open by our work
is whether the factor of $1-\frac1e$
for the general case of $n$ boxes can
be improved. We conjecture that the answer
is yes. In fact, we believe it is plausible
that the ratio between the values of the
optimal  committing policy and the optimal
adaptive policy is never less than $\frac45$,
even when the number of boxes is greater
than 2.

In the remainder of this section, we briefly
discuss the techniques used to achieve these
results. Our proof that  committing policies
attain a $(1-\frac1e)$-approximation to the
value of the optimal adaptive policy starts
with a crucial observation: Pandora's
problem with nonobligatory inspection
can be recast as an equivalent problem
in which inspection is obligatory, and the
boxes are grouped into pairs each
consisting of one of the boxes from the
original (nonobligatory) problem instance
paired with a ``doppelganger'' whose
inspection cost is zero and whose value
is deterministically equal to the
expected value of the first box.
To make the problem with paired boxes
equivalent to the original problem instance,
we must impose an additional
constraint that search policies for the
paired-box problem may open at
most one of the two boxes in each pair.
This reduction, which appears simple and
natural in hindsight, is crucial
because it enables the application of two
powerful tools. The first is a lemma of
Kleinberg, Waggoner, and Weyl~\cite{kww}
that reduces the analysis of policies for
Pandora's problem and generalizations to
the analysis of algorithms for the
same optimization problem when the values
of items are revealed for free, but
are sampled from modified distributions.
In \Cref{sec:weitzman} we generalize
the lemma to account for policies that
may select a box without opening it, a
generalization which is vital for our
application.
The second tool is a theorem of Asadpour and
Nazerzadeh about the adaptivity gap of
stochastic submodular function maximization
problems. Once Pandora's problem
with nonobligatory inspection has been transformed
into a form where these two ingredients apply,
the derivation of the
$(1-\frac1e)$-approximation result
becomes nearly automatic.
The combination of the two
ingredients --- the Kleinberg-Waggoner-Weyl
amortization lemma from~\cite{kww} together
with adaptivity gaps for stochastic probing ---
was pioneered by Singla~\cite{singla}
to solve a problem he refers to as
{\em constrained utility maximization in
the price of information model}, which
generalizes our paired-box problem with
probing constraints.

To prove that the optimal  committing policy
can be identified in polynomial time we combine
three easy observations.
\begin{enumerate}
\item Of the $2^n$ ways of partitioning the $n$
boxes into those that always remain closed and those
that are never selected while closed, we need
only consider the $n+1$ partitions in which there
is at most one box of the former type.
\item For a fixed partition of boxes into two sets
as above, the optimal  committing policy constrained
to use this partition can easily be determined
by applying Weitzman's theorem.
\item The value of this constrained optimal policy
can be calculated in polynomial time.
\end{enumerate}

Finally, to show that the gap between  committing
policies and fully adaptive policies is $\frac45$
in the special case of two boxes, we express the
value of the optimal policy as a convex combination
of two quantities: its expected value conditional
on selecting a closed box, and its expected value
conditional on selecting an open box. We then
design a probability distribution over  committing
policies whose expected value can be bounded below
by a weighted sum of the same two quantities.
Minimizing the ratio of these two weighted sums
boils down a question about minimizing a specific
bivariate function, which can be solved by direct
calculation.

\subsection{Related Work}

We have already discussed the foundational work
on optimal search theory in economics, particlarly
Weitzman's paper~\cite{weitzman} that introduced
Pandora's problem and derived its solution.
The optimality of Weitzman's procedure turns out
to be a special case of the Gittins Index
Theorem~\cite{gittins,gittins-jones}, which ironically was
proven earlier although Weitzman obtained his
results independently and the connection between
these two theorems was only realized afterward.

Doval~\cite{doval} was the first to address
Pandora's problem with nonobligatory inspection,
though special cases were anticipated
in earlier unpublished work by Postl~\cite{postl}.
In addition to examples illustrating that optimal
policies in general need to be adaptive (as described
above), Doval's main results identify sufficient
conditions for the optimal policy to have
a simple structure. In particular, Theorem 1
in~\cite{doval} identifes a sufficient
condition under which the optimal policy is
a committing policy. The sufficient condition
is quite technical, but one corollary is that
a committing policy is optimal
whenever boxes have equal inspection costs and
are totally ordered by the ``mean-preserving spread''
relation. Doval also provides a complete solution
for the case when the boxes have equal costs,
the value of each is sampled from a distribution
with two-point support, and the lower support
point is the same for all boxes.

The blending of Pandora's problem with ideas
from combinatorial optimization and algorithmic
game theory was initiated by Kleinberg, Waggoner,
and Weyl~\cite{kww}. Their paper introduced a
novel method for analyzing optimal and
approximately-optimal policies for Pandora's
problem and generalizations, by relating the
expected utility of the policy to expected values
of related quantities in a simpler environment
without inspection costs. The paper primarily
applies this method to analyze the price of
anarchy of a descending price auction when bidders
face a cost to inspect their own value, but
it also analyzes various extensions including
one in which inspection is optional; the price
of anarchy of the descending auction in this
setting is shown to be no worse than
$\frac12 - \frac1{2e} \approx 0.316$.
Singla~\cite{singla} applied the analysis
technique introduced in~\cite{kww} to
a much broader family of combinatorial
optimization problems, providing a general
transformation to convert
{\em frugal algorithms} (a type of greedy
algorithm) for combinatorial optimization
problems into policies for solving
combinatorial counterparts to Pandora's
problem, i.e., generalizations in which
the searcher still must pay a cost to
open each box, but may be allowed to
select multiple boxes, subject to
feasibility constraints on the set of
selected boxes. As noted earlier, among
the problems solved in~\cite{singla} is
a constrained utility maximization problem
featuring probing constraints that
generalize the probing constraint in our
paired-box problem. 

Adaptivity gaps have been studied for various stochastic
optimization problems. Any such problem consists of a set
of elements whose values are independent random variables.
The algorithm knows the distributions of these variables,
but not the actual realizations. The only way to learn the
actual realizations is to probe these elements. If the
value of the optimal adaptive probing policy can always
be approximated, to within a factor of $\alpha$, by the
value of a simple policy that performs probes in a fixed,
predetermined order until a stopping time is reached, then
we say the problem has an adaptivity gap of $\alpha$.
In one of the earliest papers on adaptivity gaps in
stochastic optimization,
Dean, Goemans, and Vondrak~\cite{dgv} studied a
stochastic variant of the $0/1$ knapsack problem, where
items have deterministic values but their sizes are independent
random variables and the act of placing an item in the knapsack
reveals its size. They showed that adaptivity gap is constant and
provided constant factor non-adaptive approximations.

The proof of our main result makes use of adaptivity gaps for
stochastic submodular optimization with constraints on probing.
Asadpour and Nazerzadeh~\cite{an} bound the adaptivity gap to
$1 - \frac1e$ for maximizing stochastic
monotone submodular functions when elements to probe should satisfy
matroid feasibility constraints. Adaptivity gaps for much more
general families of constraints were subsequently proven by
Gupta, Nagarajan, and Singla~\cite{gns16,gns17}.
In addition to feasibility constraints
over sets of elements to probe, there may also be constraints on the
ordering of the probes. Gupta, Nagarajan, and Singla showed a
constant adaptivity gap for submodular functions under arbitrary
prefix-closed constraints on the sequence of elements
probed~\cite{gns17}.


\section{Preliminaries}
\label{sec:prelim}

In this section, we formally define our model and discuss two related problems: \textit{Pandora's problem with required inspection} and \textit{maximizing a stochastic monotone submodular function}. Then we introduce a class of search procedures called {\em committing
policies} and explain why the optimal committing policy has a simple structure and is computationally easy to identify and implement.

\subsection{Model}\label{model}
An agent has a set of $n$ boxes. Box $i$, $1 \le i \le n$, contains a prize, $v_i$, distributed according to distribution $F_i(v_i)$ with expected value $\E v_i$. Prizes inside boxes are independently distributed.
Box $i$ has inspection cost $c_i$. While $F_i$ and $c_i$ are known; $v_i$ is not.

The agent sequentially inspects boxes, and search is with recall. Given a set of uninspected boxes, $\unopened$, and a vector of realized sampled prizes, $v$, the agent decides whether to stop or to continue search; if she decides to continue search she decides which box in $\unopened$
to inspect next. If she decides to inspect box $i$, she pays cost $c_i$ to instantaneously learn her value $v_i$.
If she decides to stop search, she can choose to select whichever box she pleases, regardless of whether it is inspected or not.
We use $\inspecti$ as an indicator for box $i$ being inspected and $\allocsubi$ as an indicator for the agent obtaining box $i$. Since one box can be obtained, $\sum_i \allocsubi \le 1$. The agent is an expected utility maximizer, where utility, $u$, is defined as the value of the box selected minus the sum of inspection costs paid. Given $v$, the vector of realized sampled prizes, and the two vectors of indicator variables, $\allocsub{ }$ and $\inspect{ }$, respectively indicating which boxes were selected and inspected, we have:

\[u(v,\allocsub{ },\inspect{ }) = \sum_i (\allocsubi v_i - \inspecti c_i) .\]

\subsection{Required Inspection}
\label{sec:weitzman}

Consider imposing the additional constraint that
a box can only be selected after it is inspected.
In other words, we require $\allocsubi \le \inspecti$ for each $i$.

Weitzman~\cite{weitzman} finds the optimal procedure to maximize expected utility when inspection is required. The optimal solution is an index-based policy, in which the agent inspects boxes in decreasing order of their indices, $\strike_i$, where
$\strike_i$ is the unique solution to
\[ {\E}_{v_i\sim F_i}\left[ (v_i - \strike_i)^+ \right] = c_i \]
and is also known as the reservation value of box $i$.
The search stops either when one of the realized values is above the
reservation value of every remaining uninspected box,
or when the agent has inspected all of the boxes.

Kleinberg et al.~\cite{kww} develop a new interpretation
of Weitzman’s characterization. They introduce an important
property of policies that we will call ``non-exposure'',
defined as follows.
\begin{definition}[Non-exposed Policy]
A policy is non-exposed if it is guaranteed to select any inspected box
whose value is found to satisfy $v_i > \strike_i$. In other words,
a policy is non-exposed if the event $(\inspecti - \allocsubi) (v_i - \strike_i) > 0$
has probability zero, for every box $i$.
\end{definition}
The key to the analysis of Weitzman's
optimal policy in~\cite{kww} is a family of random variables
$\covered_i \stackrel{\Delta}{=} \min \{ v_i, \strike_i \}$
defined for each box $i$. Kleinberg et al.\ prove that for
any policy that satisfies the required-inspection constraint
$\allocsubi \le \inspecti$, the
net contribution of box $i$ to the
expected value of the policy is
bounded above by $\E[\allocsubi \covered_i]$,
with equality if and only if the policy is non-exposed.

\begin{lemma} \label{lm:amortization} \cite{kww}
Given any $F_i$ and any policy that satisfies $\allocsubi \le \inspecti$ pointwise,
  \begin{equation} \label{eq:kww}
    \E \left[ \allocsubi v_i - \inspecti c_i \right] \le \E [\allocsubi \covered_i] .
  \end{equation}
Furthermore, this holds with equality for every box $i$
if and only if the policy is non-exposed.
\end{lemma}

\Cref{lm:amortization} can be interpreted as providing an
accounting scheme that amortizes a policy's expected inspection costs
by deducting them from the expected value of the box it eventually selects.
This accounting scheme exactly characterizes the value of non-exposed
policies, and furnishes an upper bound on the value of every other policy.
The benefit of the amortization is that it reduces the problem of analyzing
policies for Pandora's problem to the (generally simpler) problem of
analyzing rules for selecting boxes in an environment where the value
of box $i$ is $\kappa_i$, and this value can be queried at no cost.
A first application of this technique is the following
characterization of the optimal policy with required
inspection, and its expected utility.
\begin{corollary} \label{cr:w_utility} \cite{kww}
  Weitzman's policy on boxes $1 \le i \le n$ with distributions $F_i$ and
  inspection costs $c_i$, achieves expected utility $\E[\max_i \covered_i]$;
  the expected utility of any other policy cannot exceed this bound.
\end{corollary}

Since Pandora's problem with nonobligatory inspection
allows policies that may violate the inequality
$\allocsubi \le \inspecti$, in the sequel
we will need a generalization of \Cref{lm:amortization}
that pertains to such policies.
\begin{lemma} \label{lm:new-amort}
  Given any policy for Pandora's problem with nonobligatory
  inspection, and any box $i$, let
  \[
    \tilde{\covered}_i = \begin{cases}
      \covered_i & \mbox{if } \inspecti = 1 \\
      \E v_i & \mbox{if }  \inspecti = 0 .
    \end{cases}
  \]
  The inequality
  \begin{equation} \label{eq:amort}
    \E \left[ \allocsubi v_i - \inspecti c_i \right] \le \E [\allocsubi \tilde{\covered}_i] .
  \end{equation}
  is always satisfied, and the two sides are equal for every box $i$
  if and only if the policy is non-exposed.
\end{lemma}
\begin{proof}
First observe that $v_i$ is independent of $\inspecti$,
hence
\begin{align}
  \label{eq:indep.1}
    \E[ v_i \, \mid \, \inspecti ] &= \E v_i \\
  \label{eq:indep.2}
    \E[ (v_i - \strike_i)^+ - c_i \, \mid \, \inspecti ] &=
    \E[ (v_i - \strike_i)^+ - c_i ] = 0 .
\end{align}
Both of these equations will be used in the sequel.

To prove the inequality asserted in the lemma,
we will prove
the following inequality of conditional
expectations pointwise, then integrate over
$\inspecti$.
\begin{equation} \label{eq:amort.2}
  \E[ \allocsubi v_i - \inspecti c_i \, \mid \, \inspecti ]
  \le
  \E[ \allocsubi \tilde{\covered}_i \, \mid \, \inspecti ]
\end{equation}
There are two cases to consider. When $\inspecti = 0$,
$\allocsubi$ is conditionally independent of
$v_i$ because the contents of box $i$ are never even
inspected, so can have no influence on the decision
whether to select box $i$ or not.
Hence
\begin{align*}
  \E[ \allocsubi v_i - \inspecti c_i \, \mid \, \inspecti = 0 ] &=
  \E[ \allocsubi v_i \, \mid \, \inspecti = 0 ] \\
  &=
  \E[ \allocsubi \, \mid \, \inspecti = 0 ] \cdot
  \E[ v_i \, \mid \, \inspecti = 0 ] \\
  &=
  \E[ \allocsubi \, \mid \, \inspecti = 0 ] \cdot (\E v_i) \\
  &=
  \E[ \allocsubi \tilde{\covered}_i \, \mid \, \inspecti = 0 ]
\end{align*}
which establishes that the integrands on the two sides
of inequality~\eqref{eq:amort.2} are equal when
$\inspecti=0$. When $\inspecti=1$ we use the equation
$\tilde{\covered}_i = \covered_i = v_i - (v_i - \strike_i)^+$
in the following manipulation.
\begin{align*}
  \E[ \allocsubi v_i - \inspecti c_i \, \mid \, \inspecti = 1 ] &=
  \E[ \allocsubi \tilde{\covered}_i +
      \allocsubi (v_i - \strike_i)^+ - c_i \, \mid \, \inspecti = 1 ] \\
  & \le
  \E[ \allocsubi \tilde{\covered}_i + (v_i - \strike_i)^+ - c_i \, \mid \, \inspecti = 1 ] \\
  & =
  \E[ \allocsubi \tilde{\covered}_i \, \mid \, \inspecti=1 ]
\end{align*}
Hence inequality~\eqref{eq:amort.2} also holds when $\inspecti=1$.

The final sentence of the lemma asserts a necessary and sufficient
condition for equality in~\eqref{eq:amort}. To justify this
condition, note that every step in the
derivation of inequality~\eqref{eq:amort}
is an equation except for the inequality
\begin{equation} \label{eq:amort.3}
    \E[ \allocsubi \tilde{\covered}_i +
    \allocsubi (v_i - \strike_i)^+ - c_i \, \mid \, \inspecti = 1 ] \\
      \le
    \E[ \allocsubi \tilde{\covered}_i + (v_i - \strike_i)^+ - c_i \, \mid \, \inspecti = 1 ] .
\end{equation}
Hence, strict inequality holds in~\eqref{eq:amort}
if and only if there is a
positive probability that $\inspecti=1$ and
$\allocsubi (v_i-\strike_i)^+ < (v_i - \strike_i)^+$.
The relations $\inspecti=1$ and
$\allocsubi (v_i-\strike_i)^+ < (v_i - \strike_i)^+$
hold precisely when the policy violates the definition
of non-exposure.
\end{proof}

\subsection{Stochastic Submodular Maximization}
\label{sec:stochastic}

Consider the problem of maximizing a stochastic monotone submodular function $f$ with respect to a matroid
constraint $\M$. Suppose $f: \R^n_+ \rightarrow \R_+$ is a function of $n$ random variables, namely, $\A = \{X_1, X_2, \cdots, X_n\}$.
Assume $f$ is submodular, meaning
\begin{equation} \label{eq:submodular}
    \forall x,y \in \R^n_+ \quad f(x) + f(y) \ge f(x \wedge y) + f(x \vee y)
\end{equation}
where $x \wedge y$ and $x \vee y$ respectively denote the coordinate-wise
minimum and maximum of vectors $x$ and $y$.

A policy $\pi$ picks the elements to inspect one by one (perhaps, based on the realized value of the previous elements) until it stops. Once $\pi$ stops, the current state is a random vector $\Theta^{\pi} = (\theta_1, \theta_2, \cdots, \theta_n)$, where $\theta_j$ denotes the realization of $X_i$, if $i$ is inspected by the policy, and is equal to $0$ otherwise.
The objective of stochastic submodular maximization is to optimize the expected value of a policy, i.e., $\underset{\pi}{\textrm{Maximize}}~\E[f(\Theta^{\pi})]$, subject to feasibility.
The feasibility constraint is modeled using a matroid. For a given matroid $\M$ defined on the ground set of the aforementioned random variable set $\A$, a policy $\pi$ is called feasible if the subset of random variables it inspects is always an independent set of $\M$.

Asadpour and Nazerzadeh \cite{an} compare the performance of the best \textbf{adaptive} and \textbf{non-adaptive} policies. In adaptive policies, at each point in time all the information regarding the previous inspections of the policy is known. In other words, the policy has access to the actual realized value of all the elements it has inspected so far. In contrast, non-adaptive policies do not have access to such information and should make their decisions (about which random variables to inspect) before observing the outcome of any of them. They show that there exists a non-adaptive policy that achieves at least a $1- \frac1e \approx 0.63$ fraction of the value of the optimal adaptive policy.

\begin{lemma}\label{lm:AN}\cite{an}
There exists a non-adaptive policy that achieves $1- \frac1e \approx 0.63$ fraction of the optimal policy in maximizing a stochastic monotone submodular function with respect to matroid feasibility.
\end{lemma}

We now use the multilinear relaxation of $f$ to define the value of {\em fractional non-adaptive policies} \cite{an}.  A fractional non-adaptive policy is determined by a vector $y \in [0,1]^n$. This policy inspects elements in the (random) set $Y$, a set that is defined to include each $X_i\in\A$ with probability $y_i$, independently for each $i$.

We use $F(y)$ to denote the expected value obtained by the fractional non-adaptive policy associated with $y$.
Using the notation $\Theta^Y$ to denote the random vector $\Theta^\pi$ when $\pi$ is the non-adaptive
policy associated with set $Y$, we have
\begin{eqnarray} \label{multi-linear}
 F(y) :=  \sum_{Y \subseteq \{0, 1\}^n} \left[\left(\prod_{i \in Y} y_i \prod_{i \notin Y} (1 - y_i) \right) \E f(\Theta^Y)\right].
\end{eqnarray}

\begin{lemma} \label{lem:fractional_non_adapt}
\cite{an} For any monotone submodular function with matroid $\M$ feasibility constraint, for any $y$ in the base polytope of  $\M$, there exists an integral (deterministic) non-adaptive policy with expected value greater than or equal to $F(y)$.
\end{lemma}

\subsection{Committing Policies with Nonobligatory Inspection}
\label{sec:committing}

Consider the problem of maximizing expected utility for the box problem with
nonobligatory inspection (as discussed in \Cref{model}). A class of policies
that will be central to our analysis are the {\em committing policies}, which
were discussed in \Cref{sec:intro} and are defined formally as follows.

\begin{definition}[Committing Policy]
  A policy is called {\em committing} if there exists a partition of the $n$ boxes
  into two sets, $S$ and $T$, and a total ordering of the elements of $T$,
  denoted by $\prec$, such that the following properties hold.
  \begin{enumerate}
  \item The policy never inspects a box in $S$:
    $\forall i \in S \;\; \E[\inspecti] = 0.$
  \item The policy never selects a box in $T$ before inspecting it:
    $\forall j \in T \;\; \E[\allocsub{j} \cdot (1 - \inspect{j})] = 0.$
  \item If $j,k \in T$ and $j \prec k$ then the policy never inspects $k$
    before it has inspected $j$.
  \end{enumerate}
  The set $S$ is called the {\em reservation set} of the committing
  policy.
\end{definition}

Among committing policies with a fixed reservation set, $S$, it is
easy to identify the one that maximizes expected utility.
\begin{definition} \label{def:P_S}
  Policy $\Pols$ simulates running Weitzman's optimal policy on a
  modified set of boxes, in which the boxes in $T = [n] \setminus S$
  are unchanged, but each box in $i \in S$ is modified so that its
  inspection cost is zero, and its value distribution is a point mass
  on $\E v_i$. When the policy in the simulation inspects or selects
  a box in $T$, policy $\Pols$ performs the same operation. When it
  inspects a box in $S$, policy $\Pols$ instead selects the same box
  without inspecting it.
\end{definition}
The proof of the following lemma is easy, and we defer it
to \Cref{sec:committing-omitted}, along with the
(also easy) proofs of the remaining two lemmas in this section.
\begin{lemma} \label{lm:pols}
  For every $S \subseteq [n]$, policy $\Pols$ attains the highest
  expected utility among
  all committing policies with reservation set $S$.
\end{lemma}
According to \Cref{lm:pols}, the optimal committing
policy must be one of the $2^n$ elements in the set
$\{\Pols \, : \, S \subseteq [n] \}$. In fact, it is
easy to see that the optimal committing policy must
belong to a much smaller set with just $n+1$ elements.
Define $\W$ to be the Weitzman's optimal policy
on the given (unmodified) set of $n$ boxes;
equivalently $\W = \Pol{\emptyset}$.
Also, for $i \in [n]$, define $\Poli = \Pol{\{i\}}$
to be the optimal committing policy with reservation
set $\{i\}$.
\begin{lemma} \label{lm:poli}
  The optimal committing policy always belongs to the
  set $\{\W,\Pol{1},\Pol{2},\ldots,\Pol{n}\}$.
\end{lemma}
\begin{lemma} \label{lm:compute-value}
  For any $S \subseteq [n]$, the expected utility
  of policy $\Pols$ can be computed in time
  $\operatorname{poly}(n,s)$, where $s$ is
  the maximum number of support points in any
  of the distributions $F_i$.
\end{lemma}
Therefore, one can identify the optimal committing policy
in polynomial time by evaluating the expected utility
of each policy in the set $\{\W,\Pol{1},\Pol{2},\ldots,\Pol{n}\}$
and selecting the best of these $n+1$ alternatives.

\section{$1-\frac1e$ Approximation}\label{general}

In this section we analyze the worst-case ratio between
the value of the optimal committing policy and that of
the optimal policy.

\begin{theorem}\label{thm:general}
At least one of policies $\W$ and $\Poli$, $1 \le i \le n$,
achieves at least $1-\frac1e \approx 0.63$ of the optimal utility
for the box problem with nonobligatory inspection.
\end{theorem}


We establish a correspondence between the box problem and stochastic submodular
optimization. Recall from \Cref{sec:stochastic} that an instance of stochastic
submodular optimization is specified by a set of random variables
$\A = \{X_1, X_2, \cdots, X_m\}$, a submodular function
$f : \R^m_+ \to \R_+$, and a matroid $\M$ with ground set $\A$.
A policy $\pi$ chooses (either adaptively or non-adaptively)
a subset $I \subseteq \A$ of random variables whose values it
probes, subject to the constraint that $I$ must be an independent set
in $\M$. The value obtained when running policy $\pi$ is
the random variable $f(\Theta^\pi)$, where $\Theta^\pi$
denotes the random vector $(\theta_1,\ldots,\theta_m)$
specified by setting $\theta_i = X_i$ if $i \in I$ and
$\theta_i=0$ otherwise.

\begin{definition}\label{df:correspondence}[Associated Stochastic Optimization Problem]
Given an instance of Pandora's problem with nonobligatory inspection,
having $n$ boxes with costs $c_i$ and values $v_i \sim F_i$, the
{\em associated stochastic optimization problem} has $m = 2n$ random
variables denoted by
$$\A = \{X_{1,0}, X_{1,1}, X_{2,0}, X_{2,1}, \cdots, X_{n,0}, X_{n,1} \},$$
submodular objective function
\[
  f(\theta_{1,0},\theta_{1,1},\ldots,\theta_{n,0},\theta_{n,1})=
  \max \{ \theta_{i,j} \,:\, 1 \le i \le n, 0 \le j \le 1\},
\]
and matroid constraint $\M$ defined by the partition matroid whose
independent sets are all the subsets of $\A$ that contain at most
one element of each pair $\{X_{i,0}, X_{i,1}\}_{i=1}^{n}$.
The distributions of the random variables are defined as follows:
$X_{i,0}$ is drawn from the same distribution as $\covered_i$, whereas
$X_{i,1}$ is deterministically equal to $\E v_i$.
\end{definition}

Probing the first element of pair $(X_{i,0}, X_{i,1})$ in the
associated stochastic optimization problem corresponds to
inspecting box $i$ in the box problem. Probing the second
element of the pair corresponds to selecting box $i$ uninspected.
This correspondence is formalized by the following pair of
policy transformations.

\begin{definition}\label{df:policy-trans}
  Let $\mathcal{I}$ denote an instance of Pandora's problem
  with nonobligatory inspection, and let $\mathcal{J}$ denote
  the associated stochastic optimization problem.

  If $\pi$ is any (possibly adaptive)
  policy for Pandora's problem $\mathcal{I}$
  let $\Phi(\pi)$ denote the adaptive policy
  for $\mathcal{J}$ that simulates $\pi$ running
  in $\mathcal{I}$ and
  performs the following sequence of probes:
  whenever $\pi$ inspects box $i$, $\Phi(\pi)$
  probes $X_{i,0}$, and whenever $\pi$ stops and
  selects any box, $\Phi(\pi)$ probes every
  variable in the set $\{ X_{j,1} \, : \, j \in \unopened\}$,
  where $\unopened$ denotes the set of boxes
  in $\mathcal{I}$ that were uninspected at the moment
  when $\pi$ stopped.

  If $\rho$ is a non-adaptive policy for stochastic
  optimization problem $\mathcal{J}$
  and $B(\rho) \subset \A$ is the set of random variables
  that $\rho$ probes, let $S(\rho)$ denote the set of boxes
  $\{i \mid X_{i,0} \in B(\rho) \}$ and let $\Psi(\rho)$ denote the
  committing policy $\Pol{S(\rho)}$ for Pandora's problem
  $\mathcal{I}$.
\end{definition}

In the following lemmas, as in the preceding definition,
$\mathcal{I}$ denotes an instance of Pandora's problem
with nonobligatory inspection and $\mathcal{J}$ denotes
its associated stochastic optimization problem.
If $\pi$ is a policy for either problem $\mathcal{I}$ or
$\mathcal{J}$, we will use the notation $u(\pi)$ to
denote the expected utility of running policy $\pi$.
In the case of Pandora's problem this means
$u(\pi) = \E \left[ \sum_i (\allocsubi v_i - \inspecti c_i) \right]$.
In the case of the associated stochastic optimization
problem it means $u(\pi) = \E \left[ f(\Theta^{\pi}) \right]$.

\begin{lemma}\label{lm:nonadapt}
  If $\rho$ is a non-adaptive policy for $\mathcal{J}$
  and $\Psi(\rho)$ is the corresponding committing policy
  for $\mathcal{I}$, then
  \begin{equation} \label{eq:nonadapt}
    \max \{ u(\W), u(\Pol{1}), \cdots, u(\Pol{n}) \}
    \; \ge \;
    u(\Psi(\rho))
    \; \ge \;
    u(\rho) .
  \end{equation}
\end{lemma}
\begin{proof}
  Since $\Psi(\rho)$ is a committing policy, the
  inequality $\max \{ u(\W), u(\Pol{1}), \cdots, u(\Pol{n}) \}
  \ge u(\Psi(\rho))$ follows directly from
  \Cref{lm:poli}, so we focus on the inequality
  $u(\Psi(\rho)) \ge u(\rho)$ for the remainder
  of the proof.

  Couple the probability spaces of the two optimization
  problems such that when the prize inside box $i$ is
  $v_i$, the value of random variable $X_{i,0}$ equals
  $\covered_i = \min\{v_i, \strike_i\}$. Note that such
  a coupling exists, because the random variables
  $\{X_{i,0}\}_{i=1}^n$ are mutually independent and
  $X_{i,0}$ has the same marginal distribution as
  $\covered_i$ by construction.

  By construction, policy $\Psi(\rho) = \Pol{S(\rho)}$
  is non-exposed. According to \Cref{lm:amortization},
  then,
    \begin{equation}
      u(\Psi(\rho)) = \E[\max_i \tilde{\covered}_i],
      \label{eq:psi.1}
    \end{equation}
  where
  $\tilde{\covered}_i = \covered_i$ if $i \in S(\rho)$
  and $\tilde{\covered}_i = \E v_i$ if $i \not\in S(\rho)$.
  As for $u(\rho) = \E[f(\Theta^{\rho})]$, by the definition
  of $f$ and of $\Theta^{\rho}$ we have
    \begin{equation}
      u(\rho) = \E[\max_i \tilde{\theta}_i]
      \label{eq:psi.2}
    \end{equation}
  where $\tilde{\theta}_i = X_{i,0} = \covered_i$
  if $X_{i,0} \in B(\rho)$, $\tilde{\theta}_i = X_{i,1}
  = \E v_i$ if $X_{i,1} \in B(\rho)$, and
  $\tilde{\theta}_i = 0$ otherwise. In the
  former two cases $\tilde{\covered}_i = \tilde{\theta}_i$
  whereas in the third case $\tilde{\covered}_i \ge 0 = \tilde{\theta}_i$.
  Hence $\tilde{\covered}_i \ge \tilde{\theta}_i$ pointwise.
  Combining this inequality with~\eqref{eq:psi.1}-\eqref{eq:psi.2}
  and using the fact that the random variables $\{\tilde{\covered}_i\}_{i=1}^n$
  are mutually independent, as are $\{\tilde{\theta}_i\}_{i=1}^n$,
  the inequality $u(\Psi(\rho)) \ge u(\rho)$ follows.
%
%
\end{proof}
%

\begin{lemma}\label{lm:adapt}
  If $\pi$ is any (possibly adaptive) policy for Pandora's problem $\mathcal{I}$,
  and $\Phi(\pi)$ is the corresponding policy for the associated stochastic
  optimization problem, then $u(\Phi(\pi)) \ge u(\pi)$.
\end{lemma}
\begin{proof}
  As in the proof of \Cref{lm:nonadapt}, couple the probability spaces
  of the two optimization problems such that the value of the random
  variable $X_{i,0}$ equals $\covered_i = \min\{v_i,\strike_i\}$.
  By construction of policy $\Phi(\pi)$, the set of random variables
  it probes is
  $\{ X_{i,0} \mid \inspecti = 1 \} \cup \{ X_{i,1} \mid \inspecti = 0 \}$.
  Hence, if we define $\tilde{\covered}_i =
  \covered_i$ when $\inspecti=1$ and $\tilde{\covered}_i = \E v_i$
  when $\inspecti=0$, then we have
  \begin{equation} \label{eq:phi.0}
    u(\Phi(\pi)) = \E [ \max_i \tilde{\covered}_i ] \ge
    \sum_i \E [ \allocsubi \tilde{\covered}_i ] .
  \end{equation}
  \Cref{lm:new-amort} implies the following upper bound
  on $u(\pi)$.
  \begin{equation} \label{eq:phi.1}
    u(\pi) = \E \left[ \sum_i (\allocsubi v_i - \inspecti c_i) \right]
      \le \E \left[ \sum_i \allocsubi \tilde{\covered}_i \right]
  \end{equation}
  Combining this relation with inequality~\eqref{eq:phi.0}
  completes the proof.
\end{proof}

\begin{proof} [Proof of Theorem \ref{thm:general}]
  If $\pi$ denotes the optimal policy for an instance
  $\mathcal{I}$ of Pandora's problem with nonobligatory
  inspection, and $\mathcal{J}$ denotes the associated
  stochastic optimization problem, let $\rho$ denote
  an optimal non-adaptive policy for $\mathcal{J}$.
  We have the chain of inequalities
  \[
    \max \{ u(\W), u(\Pol{1}), \cdots, u(\Pol{n}) \}
    \ge
    u(\rho)
    \ge
    \left( 1 - \tfrac1e \right) \cdot
    u(\Phi(\pi))
    \ge
    \left( 1 - \tfrac1e \right) u(\pi)
  \]
  where the first inequality is \Cref{lm:nonadapt},
  the second is \Cref{lm:AN}, and the third
  is \Cref{lm:adapt}.
\end{proof}

\section{4/5 Approximation for Two Boxes}\label{sec:two_box}
In this section we show that for the case of two boxes, $n=2$, the best of policies $\W, \Pol1, \Pol2$ achieves at least $4/5$ utility of the optimal policy. We also provide a tight example for the approximation factor.

\begin{theorem}\label{thm:two_box}
At least one of policies $\W, \Pol1, \Pol2$ achieves at least $4/5$ utility of the optimal policy for the box problem with nonobligatory inspection in a setting with two boxes. This approximation factor is tight.
\end{theorem}

The proof supplies an upper bound on the optimal value by characterizing the optimal policy in the two-box case. Using ideas similar to those of Asadpour and Nazerzadeh~\cite{an}, given the optimal policy we consider a corresponding fractional non-adaptive policy. By comparing the better of the fractional non-adaptive policy and a policy that leaves all boxes uninspected, with the optimal policy we show that $4/5$ of the optimal is achievable.

\subsection*{Optimal Policy Characterization and Evaluation}

We first characterize the potential optimal policies in a problem with two boxes.
The following lemma summarizes some trivial observations, hence its proof is omitted.
\begin{lemma}
The optimal policy in the two-box problem with nonobligatory inspection falls into
one of three categories:
\begin{enumerate}
\item \label{case_open} it always selects an open box;
\item \label{case_closed} it always selects a closed box;
\item \label{case_sometimes} it sometimes selects an open
  box and sometimes a closed box.
\end{enumerate}
In case \ref{case_open}, the policy is equivalent to $\W$ with expected
utility equal to $\max_i \covered_i$.\\
In case  \ref{case_closed}, the expected utility equals
$\max_i \E v_i$. Suppose the equality holds for index $j$.
In this case $\max_i \E v_i \le \Pol{j}$.
\end{lemma}

The best of $\W, \Pol1$ and $\Pol2$ achieves the optimal value in cases \ref{case_open} and \ref{case_closed}. Therefore we only need to show that the approximation holds for case \ref{case_sometimes} where the optimal policy starts with inspecting a box. Without loss of generality, suppose that the optimal policy starts with inspecting box $1$.

\begin{lemma}\label{lm:threshold}
If the optimal policy starts with inspecting box $1$, it selects box $2$ without inspecting it only if $\covered_1$ is less than threshold $t$ where $t$ is the solution to $\E v_2 = \E \max\{t, \covered_2\}$.
\end{lemma}
\begin{proof}
Consider the realized $\covered_1$. The agent has the option to choose between $\E v_2$ (the value of selecting box $2$ without inspecting it) and $\E \max(\covered_1, \covered_2)$ (the value of emulating Weitzman's policy). To maximize the expected value, $\E v_2$ is chosen only if $\kappa_1 \le t$.
\end{proof}

Let $y = \Prob (\covered_1 \ge t) $. The optimal policy achieves utility $\E \left[ \max \{ \covered_1 , \covered_2 \} \cdot \indic{\covered_1\ge t} \right] + (1-y) \E v_2$.

\begin{lemma}
In the optimal policy that starts with inspecting box $1$ and
selects uninspected open box 2 with probability $1-y$,
the expected utility achieved is
$\E \left[ \max \{ \covered_1 , \covered_2 \} \cdot \indic{\covered_1\ge t} \right] + (1-y) \E v_2$.
Let $\covered'_1$ be a random variable distributed according to the conditional distribution of
$\covered_1$ given the event $\covered_1 \ge t$. Then the expected utility is
\begin{align}\label{opt_formula}
OPT = y \E \max \{ \covered'_1, \covered_2 \} + (1-y) \E v_2 .
\end{align}
\end{lemma}

\subsection*{Lower Bound on the Optimal Non-adaptive Policy}

Let $\NonAdapt$ be a fractional non-adaptive policy (defined in \Cref{multi-linear}) that inspects each element with the marginal probabilities of its inspection in the optimal policy. For our case, in pair $1$, the first element is inspected with probability $1$ and the second element with probability $0$. In pair $2$, the first element is inspected with probability $y$ and the second element with probability $1-y$. Since the probability of inspection of elements of each pair sums to $1$, $\NonAdapt$ belongs to the base polytope
of the partition matroid.

Consider a modified random variable for the first element of pair $1$ with a dominated distribution. Let this random variable be $0$ with probability $1-y$ and $\covered'_1$ with probability $y$, where $\covered'_1$ is a random variable distributed according to the conditional distribution of
$\covered_1$ given the event $\covered_1 \ge t$.
Due to the independence of random variables in non-adaptive policies and the monotonicity of maximization, this modification results in a fractional non-adaptive policy with a (weakly) lower value. Since value $0$ has no effect in maximizing non-negative numbers, we can consider the following modified realizations for our lower bound on the fractional non-adaptive policy: random variables $\covered'_1$, $\covered_2$ and $\E v_2$ are inspected with probabilities $y$, $y$, and $1-y$ respectively.

By Formula \eqref{multi-linear}, for the expected value of $\NonAdapt$ we have:
\begin{align*}
\NonAdapt(y) \ge &(1-y)^3 \E v_2 \\
+ &y^2(1-y) \E\covered'_1 \\
+ &y^2(1-y) \E\covered_2 \\
+ &y^3 \E\max\{\covered'_1, \covered_2\} \\
+ & y(1-y)^2 \E\max\{\E v_2, \covered'_1\} \\
+ & y(1-y)^2 \E\max\{\E v_2, \covered_2\} \\
+ & y^2(1-y) \E\max\{\E v_2, \covered'_1, \covered_2\}. \\
\end{align*}

Using $\E\covered'_1  + \E\covered_2 \ge \E\max \{\covered'_1, \covered_2\}$, for the first four terms we have:
\begin{align*}
(1-y)^3 \E v_2
+ y^2(1-y) \E\covered'_1
+ y^2(1-y) \E\covered_2
+ y^3 \E\max \{\covered'_1, \covered_2\}
 \ge (1-y)^3\E(v_2) + y^2 \E \max \{\covered'_1, \covered_2\}.
\end{align*}

Using the same argument, for the last three terms we have:
\begin{align*}
  y(1-y)^2 \E\max\{\E v_2, \covered'_1\}
+  y(1-y)^2 \E\max\{\E v_2, \covered_2\}
+  y^2(1-y) \E\max\{\E v_2, \covered'_1, \covered_2&\} \\
\ge [y(1-y)+y(1-y)^2]\E v_2 + y(1-y)\E[ \max\{\covered'_1, \covered_2\}-\E v_2&]^+ \\
\ge y(1-y)^2\E v_2 + y(1-y) \E\max\{\covered'_1, \covered_2&\}.
\end{align*}

Therefore
\begin{align}\label{ineq_1}
\NonAdapt(y) \ge (1-y)^2 \E v_2 + y\E\max\{\covered'_1,\covered_2\}.
\end{align}

Another valid non-adaptive policy is $\Pol{2}$ with value at least $ \E v_2$.
\begin{align}\label{ineq_2}
\Pol{2} \ge \E v_2
\end{align}

Inequalities \ref{ineq_1}, \ref{ineq_2} and \Cref{lm:poli} imply:
$$\max\{u(\W), u(\Pol{1}), u(\Pol{2})\} \ge \max\{\E v_2, (1-y)^2 \E v_2 + y\E\max(\covered'_1,\covered_2)\}$$

\subsection*{Comparing the Optimal Adaptive and Non-Adaptive Policies}

We compare the lower bound on the optimal non-adaptive policy,
$\max \{ \E v_2, (1-y)^2 \E v_2 + y\E\max(\covered'_1,\covered_2) \},$
with the utility of the optimal policy from \Cref{opt_formula},
$(1-y)\E v_2 + y\E\max(\covered'_1,\covered_2)$,
and show that the ratio is at least $\frac{4}{5}$. Note that by \Cref{lm:threshold}, $\E\max(\covered'_1,\covered_2) \ge \E v_2$.
Let
$\E\max(\covered'_1,\covered_2) =a \E v_2$
where $a \ge 1$.
We have:
  \begin{align}
	\frac{\max \{ \E v_2, (1-y)^2 \E v_2 + y\E\max(\covered'_1,\covered_2) \}}{(1-y)\E v_2 + y\E\max(\covered'_1,\covered_2)}  &=\\
    \frac{\max(1,(1-y)^2+ay)}{1-y+ay} &\ge
    \begin{cases}
      \frac{1}{1-y+ay}, & \text{if}\ a \le 2-y \\
      \\
      \frac{(1-y)^2+ay}{1-y+ay}, &  \text{if}\ a \ge 2-y.
    \end{cases}
  \end{align}

The formula for the first part is decreasing in $a$ for a fixed $y$ and achieves its minimum at $a=2-y$. The formula for the second part is increasing $a$ with fixed $y$ and therefore achieves its minimum at $a=2-y$. Therefore the maximum ratio occurs at $a=2-y$ and is equal to:
\[ \frac{1}{1-y+(2-y)y} = \frac{1}{1+y(1-y)}.\]
Since $0 \le y \le 1$,
$$\frac{1}{1+y(1-y)}\le \frac{1}{1+\frac{1}{4}}=\frac45.$$
This concludes the proof of \Cref{thm:two_box}.

The following is a tight example for \Cref{thm:two_box}.
\begin{example}
Consider boxes A and B.
Suppose box A has value $0$ with probability $\frac12$ and value $1$ with probability $\frac12$, and its inspection cost is $0$.
Box B has value $0$ with probability $1-\frac1N$ and value $N$ with probability $\frac1N$; and its inspection cost is $\frac{N-1}{2N}$.

The optimal policy starts with inspecting box A, and if the value is 0, selects uninspected box B. If the value of box A is 1, the optimal policy inspects box B and takes the maximum value of the two boxes. The expected utility of this policy is
\begin{align*}
& \frac12 \cdot \frac1N\cdot N +\frac12  \left( - \frac{N-1}{2N}+ \frac1N\cdot N + \left(1-\frac1N \right) \cdot 1 \right) \\
= & \frac12+\frac12  \left( \frac32-\frac1{2N} \right)
\end{align*}
which approaches $ \frac54 $ as $N$ goes to infinity.

Policies $\Pol{1}$, $\Pol{2}$ and $\W$ each achieve utility $1$:
Policy  $\W$, inspects both boxes and obtains the maximum value. The expected utility in this case is $- \frac{N-1}{2N} +  \frac1N\cdot N + (1-\frac1N)\cdot \frac12$.
Policy $\Pol{1}$ starts by inspecting box B. If box B has value $N$, it selects it. Otherwise it selects uninspected box A. Therefore it has utility $- \frac{N-1}{2N} +  \frac1N\cdot N + (1-\frac1N)\cdot \frac12$.
Policy $\Pol{2}$ inspects box A. If the value is 0, it selects uninspected box B. If the value of box A is 1, it is indifferent between selecting box A and uninspected box B. The expected utility in this case is $\frac12\cdot \frac1N \cdot N + \frac12 \cdot 1$.
\end{example}

\bibliographystyle{plain}
\bibliography{pandora}

\appendix
\section{Omitted Proofs Concerning Committing Policies}
\label{sec:committing-omitted}

In this section we reiterate
and prove \Cref{lm:pols,lm:poli,lm:compute-value},
which concern the structure and computation of optimal committing
policies.

\begin{lemma}[\Cref{lm:pols} restated]
  For every $S \subseteq [n]$, policy $\Pols$ attains the highest
  expected utility among
  all committing policies with reservation set $S$.
\end{lemma}
\begin{proof}
  Define a modified set of boxes as in \Cref{def:P_S}.
  Observe that in this modified problem instance,
  for any box $i \in S$ we have $\sigma_i = \E v_i$
  since $c_i=0$. Since the value of box $i$ is deterministically
  equal to $\E v_i$, whenever Weitzman's policy inspects
  box $i$ it finds that $v_i \ge \sigma_i$ and hence it
  immediately selects $i$. Thus, every execution path
  of Weitzman's policy on the modified set of boxes
  can be represented by a sequence of operations, each
  of which is either inspecting a box in $T$, selecting a
  box in $T$, or inspecting-and-immediately-selecting
  a box in $S$. Policy $\Pols$ duplicates each of these
  three types of operations and receives the same cost
  or expected benefit whenever it performs one of them,
  hence the expected utility of running Weitzman's optimal
  policy on the modified problem instance equals the
  expected utility of running $\Pols$ on the
  original instance.

  We must now show that no other committing policy with
  reservation set $S$ can attain a higher expected utility.
  This is quite easy to do, using the fact that Weitzman's
  policy is optimal for the modified instance.
  If $\pi$ is any committing policy with reservation set $S$,
  there is a corresponding policy $H(\pi)$ for the modified
  set of boxes that operates as follows: when $\pi$
  inspects or selects a box in $T$, $H(\pi)$ performs the
  same operation. When $\pi$ selects a box in $S$, $H(\pi)$
  inspects and immediately selects that box. (There is no need
  to define the behavior of $H(\pi)$ when $\pi$ inspects a box
  in $S$ since that event never happens.)
  The utility of running $H(\pi)$ on the modified set of boxes
  is the same as the utility of running $\pi$ on the original
  set of boxes, since the extra inspection operations that
  $H(\pi)$ performs on elements of $S$ have zero cost.
  Since the utility of running Weitzman's policy on the
  modified set of boxes is an upper bound on the utility
  of running $H(\pi)$, it follows that the utility of
  running $\Pols$ is an upper bound on the utility of
  running $\pi$, as claimed.
\end{proof}

\begin{lemma}[\Cref{lm:poli} restated]
  The optimal committing policy always belongs to the
  set $\{\W,\Pol{1},\Pol{2},\ldots,\Pol{n}\}$.
\end{lemma}
\begin{proof}
  Suppose $S \subseteq [n]$ is any set of two or more
  elements, and consider any two distinct elements
  $i,j \in S$ with $\E v_i \ge \E v_j$. A
  committing policy with reservation set $S$
  can never open box $i$ or box $j$, and the
  operation of selecting closed box $j$ is
  always dominated by the operation of
  selecting closed box $i$. Hence, any
  committing policy with reservation set $S$
  is dominated by a committing policy with
  reservation set $\{i\}$. In particular,
  the optimal such policy, $\Poli$, has at
  least as much expected utility as $\Pols$.
\end{proof}

\begin{lemma}[\Cref{lm:compute-value} restated]
  For any $S \subseteq [n]$, the expected utility
  of policy $\Pols$ can be computed in time
  $\operatorname{poly}(n,s)$, where $s$ is
  the maximum number of support points in any
  of the distributions $F_i$.
\end{lemma}
\begin{proof}
  Let us start with the case $S = \emptyset,
  \Pols = \W$. According to \Cref{cr:w_utility},
  the expected utility of Weitzman's optimal policy, $\W$,
  is equal to $\E[\max_i \covered_i]$.
  Let $G_i$ denote the cumulative distribution
  function of $\covered_i$, i.e.
  \[
    G_i(t) = \begin{cases}
      F_i(t) & \mbox{if } t < \sigma_i \\
      1 & \mbox{otherwise}.
    \end{cases}
  \]
  Then we have the formula
  \begin{align}
    \nonumber
    \E[\max_i \covered_i] &= \int_0^\infty
      \Pr(\max_i \covered_i > t) \, dt
    \nonumber
      = \int_0^{\infty} 1 - \Pr(\max_i \covered_i \le t) \, dt
      = \int_0^{\infty} \left( 1 - \prod_{i=1}^n G_i(t) \right) \, dt .
  \end{align}
  The integrand on the right side is a step function with
  at most $ns$ steps, since every discontinuity in the step
  function belongs to the union of the
  support sets of the distributions of $\kappa_1,\ldots,\kappa_n$.
  Hence the integral can be computed in time
  $\operatorname{poly}(n,s)$ by simply summing over the steps.

  Computing the expected utility of policy $\Pols$ in the
  general case of $S \subseteq [n]$ reduces to the special
  case $\Pols = \W$, because the expected utility of $\Pols$
  is equal to the expected utility of Weitzman's policy on
  a modified set of boxes, as was shown in the proof of
  \Cref{lm:pols}.
\end{proof}

\end{document}